\documentclass[aps,showpacs,pra,twocolumn,superscriptaddress]{revtex4-2}

\usepackage{exscale,amsfonts,mathrsfs}
\usepackage{bbm}
\usepackage{graphicx}
\usepackage{latexsym}
\usepackage{times}
\usepackage[T1]{fontenc}
\usepackage{enumerate}
\usepackage{bbold}
\usepackage{color}
\usepackage{mathtools}
\usepackage[normalem]{ulem}
\usepackage{amsfonts,amsmath,amssymb,amsthm,dsfont}
\usepackage[colorlinks=true,citecolor=blue,urlcolor=blue]{hyperref}
\usepackage{natbib}
\usepackage{physics}
\usepackage{changes}
\usepackage{titlesec}
\usepackage{palatino}
\usepackage{mathpazo}
\usepackage{subfigure}
\usepackage{mathrsfs}
\usepackage{xcolor,longtable}

\usepackage{physics}

\usetikzlibrary{patterns}

\newcommand{\one}{\leavevmode\hbox{\small1\normalsize\kern-.33em1}}

\newcommand{\ment}{\ket{\phi_+}}
\newcommand{\h}{\mathcal{H}}

\newcommand{\Q}{\mathcal{Q}}
\newcommand{\Lc}{\mathcal{L}}

\newcommand{\C}{\mathbb{C}}
\newcommand{\I}{\mathbb{1}}
\newcommand{\ti}[1]{\tilde{#1}}

\newcommand{\eg}{{\it{e.g.~}}}
\newcommand{\ie}{{{i.e.~}}}

\renewcommand{\iota}{\mathrm{i}}

\newtheorem{theorem}{Theorem}
\newtheorem{corollary}{Corollary}[theorem]
\newtheorem{lemma}[theorem]{Lemma}

\begin{document}
\title{Classifying the simplest Bell inequalities beyond qubits and their applications towards self-testing}
\author{Palash Pandya}
\email{ppandya@cft.edu.pl}
\affiliation{Center for Theoretical Physics, Polish Academy of Sciences, Aleja Lotnik\'{o}w 32/46, 02-668 Warsaw, Poland}

\author{Shubhayan Sarkar}
\affiliation{Institute of Informatics, Faculty of Mathematics, Physics and Informatics,
University of Gdansk, Wita Stwosza 57, 80-308 Gdansk, Poland}

\author{Remigiusz Augusiak}
\affiliation{Center for Theoretical Physics, Polish Academy of Sciences, Aleja Lotnik\'{o}w 32/46, 02-668 Warsaw, Poland}

\begin{abstract}
   Bell inequalities reveal the fundamentally nonlocal character of quantum mechanics. In this regard, one of the interesting problems is to explore all possible Bell inequalities that demonstrate a gap between local and nonlocal quantum behaviour. This is useful for the geometric characterisation of the set of nonlocal correlations achievable within quantum theory. Moreover, it provides a systematic way to construct Bell inequalities that are tailored to specific quantum information processing tasks. This characterisation is well understood in the simplest $(2,2,2)$ scenario, namely two parties performing two binary outcome measurements. However, beyond this setting, relatively few Bell inequalities are known, and the situation becomes particularly scarce in scenarios involving a greater number of outcomes. Here, we consider the $(2,2,3)$ scenario, or two parties performing two three-outcome measurements, and characterise all Bell inequalities that can arise from the simplest sum-of-squares decomposition and are maximally violated by the maximally entangled state of local dimension three. We then utilise them to self-test this state, along with a class of three-outcome measurements.
   \end{abstract}

\maketitle

Nonlocality and entanglement are the two most fundamental aspects of Quantum theory, and as such are the most studied topics. It is a well-known fact, first established by Bell \cite{bell-epr}, that the correlations generated by local measurements on an entangled system cannot always be reproduced by a classical theory. 
This non-reproducibility is a key ingredient in many applications bounding the length of the secret key in quantum key distribution, certifying the Hilbert space dimensions, randomness generation, secret sharing, etc. \cite{acin2007device,pironio2010random,brunner2008testing,secret-sharing}. Therefore, Bell inequalities that certify nonlocal correlations and distinguish them from classical are of utmost interest. 
Finding Bell inequalities for quantum systems and measurements of interest is in general a hard task, and apart from a few special cases, the Bell inequalities are not known. The task of finding multipartite Bell inequalities is especially difficult.
One of the most important implications of Bell inequality violation is certifying that the correlations generated by the measurement devices of the observers are nonlocal. This certification, aptly called device-independent (DI) certification, is the cornerstone of secure quantum key distribution \cite{acin2007device}. While device-independent certification reveals that the correlations produced are nonlocal, there is still an uncertainty in what is the actual quantum state shared between the observers and the measurements performed by them on their local subsystems. 

Though different sets of quantum systems and local measurements can result in the same nonlocal correlations and thus the same Bell inequality violation or 'Bell score', maximal violation of a Bell inequality can, in almost all cases, be attributed uniquely to a quantum state and a set of measurements. This is called the 'Self-testing' property and is a stricter form of device independence, where in addition to certifying nonlocal correlations, we also certify the underlying quantum state (and local measurements) \cite{mayers-yao}.
The simplest system that admits the property of self-testing is the maximally entangled state of two qubits with the now ubiquitous Clauser-Horne-Shimony-Holt (CHSH) inequality with two measurements and two outcomes per observer \cite{popescu1992states}. Since then, a complete characterisation of the Self-Tests for the maximally entangled state has been done \cite{robust-singlet,parallel-self-testing,complete-characterization-singlet}. In fact, in \cite{all-pure-states-self-test}, the authors proved that all entangled pure bipartite states can be self-tested. The methods developed for qubits were then employed for self-testing, in addition to multipartite qubit states, higher dimensional qudit systems as well, by describing the shared system locally as qubits \cite{yn13,all-pure-states-self-test,col18,supic18}.

In contrast to this, recent research has focused more on letting go of the local qubit formalism and developing the framework for self-testing in higher local dimensions by engaging the more complex operator machinery that goes with it, particularly, by expressing local measurements using the Weyl-Heisenberg formalism  \cite{Kaniewski2019maximalnonlocality,salavrakosBellInequalitiesTailored2017,sarkarCertificationIncompatibleMeasurements2022}. 
A powerful tool to derive algebraic constraints for the given Bell inequality to self-test the shared quantum state and the local measurements is the \emph{Sum-of-Squares} (SOS) decomposition  \cite{supicSelftestingQuantumSystems2020a}. By construction, the SOS decomposition of the Bell operator provides the quantum or Tsirelson bound of the Bell inequality. In \cite{custom-bell-inequalities}, this method is put forth formally as a way to design Bell inequalities tailored to the state. They compare the variational method and the sum-of-squares method for constructing Bell inequalities with the central idea of the state maximally violating the Bell inequality. The Tsirelson bound is then achieved by the state for which the Bell inequality was tailored, which is a desirable property. They point out that while the variational method provides the necessary conditions, the sum-of-squares method gives conditions on the Bell inequality that are sufficient for the state to violate it maximally.
In \cite{salavrakosBellInequalitiesTailored2017} the method is used to derive self-testing Bell inequalities, known as the SATWAP inequalities,  which are maximally violated by maximally entangled bipartite states and are defined for an arbitrary number of  measurements and outcomes per party. The observables considered were taken to be the optimal CGLMP measurements. In \cite{Kaniewski2019maximalnonlocality}, the authors address a similar question but restrict themselves to prime local Hilbert space dimension $d(\geq 3)$, and correspondingly, $d$ measurement settings and $d$ outcomes. Similarly to SATWAP inequalities, in \cite{custom-bell-inequalities} the two qutrit system is again self-tested using CGLMP measurements.

In a similar spirit, this article provides a very general analytic construction of a class of CGLMP-style Bell inequalities that self-test the bipartite maximally entangled qutrit state in the simplest setting of two three-outcome measurements. We will see that the SATWAP inequality with the number of measurement settings and outcomes fixed to 2 and 3, respectively, is a particular case of our construction. Although both share the same Tsirelson bound and a similar SOS decomposition, we do not make any assumptions about the observables used by both Alice and Bob.
The article is organised as follows. In \autoref{sec:prelim}, we give a gentle introduction to the formalism used in this article. In \autoref{sec:quantum-bound}, we define the general setting of this article, define the general Bell operator and provide a sum-of-squares decomposition. In \autoref{sec:self-test}, the Bell operator is further constrained to provide self-testing statements for the maximally entangled state. Next, in \autoref{sec:classical}, we provide an analysis of the classical value of the Bell operator and a comparison with the SATWAP inequalities.

\section{Preliminaries}
\label{sec:prelim}
Let us consider a simple \emph{Bell scenario} with two space-like separated, non-communicating observers, Alice and Bob, who share a quantum system and can perform local measurements on their devices. The scenario is characterised by the number of measurements, $m$, and the number of outcomes, $d$. To use the prevalent notation, the measurements of Alice and Bob are denoted as $A_x$ and $B_y$ where $x,y\in\{1,2,\dots, m\}$ and the outcomes are labelled as $a$ and $b$, respectively, for Alice and Bob, where $a,b\in\{0,1,\dots,d-1\}$.
In each run of the scenario, a quantum state, $\rho_{AB}$ is prepared and shared with Alice and Bob, who in turn choose a measurement setting on their device and read the outcome of the measurement. Over sufficiently many runs of the experiment, the results can be compiled into a vector of joint probability distributions, or correlations, $\vec{p}=\{p(a,b|x,y)\}$, where $p(a,b|x,y)$ is the probability of obtaining the outcomes $a$ and $b$ given that Alice and Bob perform the measurements $A_x$ and $B_y$, and is defined as,
\begin{equation}
    p(a,b|x,y) = Tr[\rho_{AB} \left(A_x\otimes B_y\right)]\,.
\end{equation}

The set of $\vec{p}$ obtained in the above manner is convex and is known as the set of quantum correlations, $\Q$. 
In such a scenario, it is not necessary that Alice and Bob's outcomes are statistically independent. In other words, the joint probability doesn't always factorise, \ie, $p(a,b|x,y) = p(a|x)p(b|y)$. The correlations such that,
\begin{equation}
    p(a,b|x,y) = p(a|x)p(b|y) \,, 
\end{equation}
where $p(a|x),p(b|y)\in\{0,1\}$ for all $x,y$ are called  \emph{deterministic} correlations. The set of correlations $\vec{p}$ that can be written as a convex sum of {deterministic} correlations is called the \emph{local} polytope, $\Lc$, and is a subset of $\Q$. 

A natural question that arises from this characterisation is whether we can determine if a given correlation $\vec{p}$ belongs to $\Q$ or $\Lc$. The way to show this is using Bell inequalities. Bell inequalities, in general, are formulated as linear combinations of the joint probabilities in $\vec{p}$,
\begin{equation}
\label{eq:bell-ineq-def}
    I(\vec{p}) = \sum_{abxy} c^{xy}_{ab} p(a,b|x,y) \leq b_C\,.
\end{equation}
where $c^{xy}_{ab}$ are the parameters of the Bell functional, and $b_C$ is called the classical bound of the Bell Inequality. It is defined as 
\begin{equation}
    b_C = \max_{\vec{p}\in\Lc}  I(\vec{p})\,.
\end{equation}
Similar to the classical bound, we can also define the quantum bound or the Tsirelson bound, $T_Q$ for the Bell inequality,
\begin{equation}
    T_Q = \max_{\vec{p}\in\Q}  I(\vec{p})\,.
\end{equation}
Using the Bell inequality now we can determine the membership of a given correlation $\vec{p}$ in the following manner. If it satisfies \eqref{eq:bell-ineq-def} it belongs in the set $\Lc$, and if $ b_C<I(\vec{p})\leq T_Q$, it belongs in the quantum set, $\Q$.

A more convenient approach is using the generalised correlators defined using the Fourier transform of the joint probabilities \cite{gen-correlators,Kaniewski2019maximalnonlocality},
\begin{equation}
    \ev{A^{(k)}_x\otimes B^{(l)}_y} = \sum_{a,b=0}^{d-1} \omega^{ak+bl} p(a,b|x,y)\,,
\end{equation}
where the $d$ measurement operators for Alice and Bob are defined as,
\begin{align}
    A^{(k)}_x = \sum_{i=0}^{d-1} \omega^{ik} P_i^A \qand B^{(l)}_y = \sum_{j=0}^{d-1} \omega^{jl} P_j^B\,,
\end{align}
where $\omega=\exp{2\pi\iota/d}$ is the $d^{\mbox{th}}$-root of unity, and $P_i^X$ are positive semi-definite operators such that $\sum_{i=0}^{d-1} P_i^X = \I$. In the case when $P_i^X$ are projective, the operators $A^{(k)}_x$ and $B{(l)}_y$ can be described by the unitary observables $A_x$ and $B_y $ such that,
\begin{equation}
    A^{(k)}_x = (A_x)^k \qand B^{(l)}_y = (B_y)^l \,.
\end{equation}
As a consequence, the spectrum of the observables is $\{1,\omega,\omega^2 \cdots \omega^{d-1}\}$. It is also evident from their construction that,
\begin{equation}
    (A_x^k)^\dagger = A_x^{d-k} \qand (B_y^k)^\dagger = B_y^{d-k}\,.
\end{equation}

In this correlation picture the Bell inequality can be expressed as
\begin{equation}
    I = \sum_{x,y=1}^{m}\sum_{k=0}^{d-1} c_{xyk}\ev{A^{k}_x\otimes B^{k}_y}\,,
\end{equation}
where the Bell scenario is defined by $m$ $d$-outcome measurements per observer.

\section{The Bell Operator and the Quantum Bound}
\label{sec:quantum-bound}
In this article, we consider the most general Bell operator without considering local terms. We want to set some expectations from our Bell operator that it must have nice properties, namely, a Sum-of-Squares decomposition and that the maximally entangled state achieves the Tsirelson bound. To that end, we define the Bell scenario as two parties performing two three-outcome measurements on their shared quantum state, so that $m=2$ and $d=3$.
We define our general Bell operator as
\begin{multline}\label{eq:BI}
    W=A_0\otimes \left(\alpha B_0+\beta B_1\right)+A_1\otimes \left(\gamma B_0+\delta B_1\right)\\
    +A_0^{\dagger}\otimes \left(\alpha B_0+\beta B_1\right)^{\dagger}+A_1^{\dagger}\otimes \left(\gamma B_0+\delta B_1\right)^{\dagger},
\end{multline}
where $A_0,A_1,B_0$ and $B_1$ are Alice and Bob's observables, and  $\alpha,\beta,\gamma,\delta\in\mathbb{C}$. We will reference them as Bell parameters going forward. Alice and Bob's observables are  $A_0,A_1,B_0$ and $B_1$. Following the discussion in the last section, the observables satisfy the following relations:
\begin{align}
 &A_x^3=B_y^3=\I, \nonumber\\
 &A_xA_x^\dagger = B_xB_x^\dagger =\I, \nonumber\\
\label{eq:operator-conditions}
 &A_x^2=A_x^\dagger \qand B_x^2=B_x^\dagger.
\end{align}

To begin the characterisation of the Bell operator, we first define the terms $L_i$,
\begin{align}
\label{eq:l1l2}
    L_1 &= \I - A_0 \otimes (\alpha B_0 + \beta B_1)\,,\\
    L_2 &= \I - A_1 \otimes (\gamma B_0 + \delta B_1)\,.
\end{align}
\begin{theorem}
\label{thm:theorem1}
If the Bell parameters satisfy $\alpha\beta^* + \gamma\delta^*=0$ and with $L_i$s defined above, then the Bell operator $W$ in \eqref{eq:BI} admits the following sum-of-squares decomposition,
\begin{multline}
    W = (2+\abs{\alpha}^2+ \abs{\beta}^2+ \abs{\gamma}^2+ \abs{\delta}^2)\I \\- \sum_{i=1}^2 \frac{1}{2}(L_i^\dagger L_i + L_i L_i^\dagger)\,,
\end{multline}
\end{theorem}

\begin{proof}
Here, we give a sketch of the proof, while the details can be found in \autoref{app:theorem1}. Expanding the sum-of-squares expression, we find
\begin{multline}
\label{eq:sos}
\frac{1}{2}\left(L_1^\dagger L_1 + L_2^\dagger L_2 + L_1 L_1^\dagger  + L_2 L_2^\dagger \right) \\= 2 \I - W +  X_1 + X_2,
\end{multline}
where $X_1$ and $X_2$ are the cross terms:
\begin{multline}
    X_1 = \I \otimes (\alpha B_0 + \beta B_1)^\dagger (\alpha B_0 + \beta B_1)\\+
    \I \otimes (\gamma B_0 + \delta B_1)^\dagger  (\gamma B_0 + \delta B_1),\,\mbox{and}
\end{multline}
\begin{multline}
    X_2 = \I \otimes (\alpha B_0 + \beta B_1)^\dagger (\alpha B_0 + \beta B_1)\\+
    \I \otimes (\gamma B_0 + \delta B_1)^\dagger  (\gamma B_0 + \delta B_1).
\end{multline}
The sum of cross terms is reduced to
\begin{equation}
    X_1+X_2 = (\abs{\alpha}^2+ \abs{\beta}^2+ \abs{\gamma}^2+ \abs{\delta}^2)) \I ,
\end{equation}
when we set
\begin{equation}
\label{eq:ab-gd-sum-zero}
    \alpha\beta^* + \gamma\delta^* = \alpha^*\beta + \gamma^*\delta=0.
\end{equation}
From this, we can rewrite \eqref{eq:sos}:
\begin{multline}
\frac{1}{2}\sum_{i=1}^2\left(L_i^\dagger L_i +  L_i L_i^\dagger \right) \\= (2 + \abs{\alpha}^2+ \abs{\beta}^2+ \abs{\gamma}^2+ \abs{\delta}^2) \I - W.
\end{multline}
Rearranging the above equation gives us the requisite sum-of-squares decomposition and completes the proof of the theorem.
Note, $T_Q = (2 + \abs{\alpha}^2+ \abs{\beta}^2+ \abs{\gamma}^2+ \abs{\delta}^2)$ is the Tsirelson bound.
\end{proof}

\begin{corollary}
It follows from the above theorem, that
\begin{enumerate}
    \item $\abs{\alpha}^2+ \abs{\beta}^2+ \abs{\gamma}^2+ \abs{\delta}^2 = 2$, and
    \item the Tsirelson bound or the maximum quantum value of $W$ is $T_Q=4$.
\end{enumerate}
\end{corollary}

\begin{proof}
For the proof of the corollary, notice that when a state $\ket\psi$ maximally violates the Bell inequality, the following relations hold:
\begin{equation}
\label{eq:nullifiers}
    L_i\ket\psi = L_i^\dagger\ket\psi = 0 \qfor  i=1,2.
\end{equation}
with $L_i$'s referred to as the nullifiers of the state. If we define $C_0 = \alpha B_0 + \beta B_1$ and $C_1 = \gamma B_0 + \delta B_1$, then using the above relations, it is easy to show that the operators $C_i$ are also unitary and projective, 
\begin{equation}
    \label{eq:C-operators}
    C_i^2 = C_i^\dagger\qand C_i^\dagger C_i =\I.
\end{equation}
Using the latter condition for both $i=0,1$ with \eqref{eq:operator-conditions} and \eqref{eq:ab-gd-sum-zero}, directly gives us
\begin{align}
\label{eq:sum-abs-vals}
    \abs{\alpha}^2+ \abs{\beta}^2+ \abs{\gamma}^2+ \abs{\delta}^2 = 2,
\end{align}
and therefore the Tsirelson bound is $T_Q = 4$. \end{proof}

\section{Self-testing statement for the Bell operator}
\label{sec:self-test}

We further characterise our Bell operator, $W$, by providing a self-testing statement for $W$, which is the main result of this article. 
The following theorem determines the state and the family of measurements that can be self-tested by our Bell operator.
First, let us define a few relevant symbols and matrices so that the mathematical expressions that follow are easier to digest. Most of the expressions that follow can be simplified if we use \eqref{eq:ab-gd-sum-zero} to define
\begin{equation}
    \kappa = \frac{\alpha\beta^*}{\alpha^*\beta} = \frac{\gamma\delta^*}{\gamma^*\delta}\,.
\end{equation}
The above is equivalent to
\begin{equation}
    \kappa ={e}^{\iota 2(\theta_\alpha-\theta_\beta)} = {e}^{\iota 2(\theta_\gamma-\theta_\delta)}\,,
\end{equation}
which constrains the phases as $\theta_\alpha-\theta_\beta = \theta_\gamma-\theta_\delta + n \pi$ for all $n\in\mathbb{Z}$.
We will also define the matrix $T_3$ in terms of the elements of the three dimensional Heisenberg-Weyl basis, $Z=\sum_{i=0}^2\omega^i\ket{i}\bra{i}$ and $X=\sum_{i=0}^2\ket{i}\bra{i+1}$, also known as the clock and shift matrices, with $\omega = e^{\frac{2\pi\iota}{3}}$ (the cube root of unity).
\begin{multline}
\label{eq:t3-def}
    T_3 = \frac{1}{3}\left[\left( \kappa^2+\frac{2}{\kappa} \right) Z + \left( \kappa^2-\frac{1}{\kappa} \right) ZX \right.\\
    \left.+ \omega\left( \kappa^2-\frac{1}{\kappa} \right) ZX^2\right].
\end{multline} 

\begin{theorem}
\label{thm:theorem2}
    If the state $\ket\psi \in \mathcal{H}_A\otimes \h_{B}$ achieves the Tsirelson bound $T_Q=4$ for the Bell operator $W$, with observables $A_x\in\h_A$ and $B_x\in\h_B$ and $x=0,1$, then the following three statements are simultaneously true:
    \begin{enumerate}
        \item There exist local unitaries $U_A:\h_A\rightarrow\mathbb{C}^3\otimes\h_{A'}$ and $U_B:\h_B\rightarrow\mathbb{C}^3\otimes\h_{B'}$, such that
        \begin{align}
        \label{eq:u_b}
            U_B B_0 U_B^\dagger = Z\otimes \I_{B'} \,,\\
            U_B B_1 U_B^\dagger = T_3\otimes \I_{B'}\,, 
        \end{align}
        and,
        \begin{align}
            U_A A_0 U_A^\dagger = \left(\alpha Z+\beta T_3\right)^*\otimes \I_{A'} \,,\\
            U_A A_1 U_A^\dagger = \left(\gamma Z + \delta T_3\right)^*\otimes \I_{A'}\,, 
        \end{align}
        where $Z$ and $T_3$ are defined in \eqref{eq:t3-def}.
        \item The action of the unitaries on the state is given by
        \begin{equation}
            U_A\otimes U_B \ket\psi = \ment\otimes\ket{\psi'}\,,      
        \end{equation}
        where $\ment$ is the maximally entangled state acting on $\C^3_A\otimes\C^3_B$ and $\ket{\psi'}$ is some state acting on $\h_{A'}\otimes\h_{B'}$.
        \item The Bell parameters are constrained as
        \begin{align}
        \label{eq:alpha-square-is-1}
            \frac{\alpha^2}{\alpha^*} + \frac{\beta^2}{\beta^*}  = \frac{\gamma^2}{\gamma^*} + \frac{\delta^2}{\delta^*} = 1.
        \end{align}
         and, 
        \begin{align}
        \label{eq:inequality}
            \left( \kappa^2-\frac{1}{\kappa} \right) \neq 0.
        \end{align}
        in addition to \eqref{eq:ab-gd-sum-zero} and \eqref{eq:sum-abs-vals}.
    \end{enumerate}
\end{theorem}
\begin{proof}
    Again, the proof is lengthy, and we will describe the main steps here, while the details are relegated to the \autoref{app:theorem2}. Without loss of generality, we can start by characterising the observables of Bob and then derive the form of Alice's observables.
    We shall first consider the case where the observables $B_0,B_1\in\C^3$. For such projective and unitary observables, it is always possible to bring $B_0 = Z$ using some unitary $U_B$ and then $B_1$ can be parameterised using the Heisenberg-Weyl basis as
    \begin{multline}
    \label{eq:B1}
        B_1=aZ+bZ^2+X(c\I+dZ+eZ^2)\\+X^2(f\I+gZ+hZ^2).
    \end{multline}
    When we substitute $B_0$ and $B_1$ in \eqref{eq:C-operators} we obtain the relations,
    \begin{equation}
        \label{eq:bell-params}
        \frac{-\alpha\beta}{\beta^2-\beta^*}=\kappa=\frac{-\gamma\delta}{\delta^2-\delta^*}\,,
    \end{equation}
    which, when rearranged, gives \eqref{eq:alpha-square-is-1}.
    For a more general construction, now consider the case where observables $A_x,B_x$ act on Hilbert spaces $\h_A$ and $\h_B$ with some general dimension $d$.
    Again using the trace of the relations \eqref{eq:C-operators}, we find that if one assumes the inequality \eqref{eq:inequality} then we can show that the observables $B_x$ are traceless. Combined with the fact that $B_x^3=\I_B$ and $B_x^\dagger B_x=\I$, tracelessness implies that the eigenvalues of the observables are $\{1,\omega,\omega^2\}$, and they occur with the same multiplicity. This in turn implies that $\h_B\rightarrow\C^3\otimes\h_{B'}$, in other words, the dimension of Bob's Hilbert space is a multiple of 3. As in the first case, we can again choose a local unitary $U'_B$ such that
    \begin{align}
        \tilde{B}_0 =& U'_B B_0 {U'_B}^\dagger = Z\otimes\I ,\\
        \tilde{B}_1 =& U'_BB_1{U'_B}^\dagger=\sum_{i,j=0}^{2}\op{i}{j}\otimes F_{ij}.
    \end{align}
    where $F_{ij}$'s are matrices that act on $\h_{B'}$. Using \eqref{eq:C-operators} and the fact that the operators $\tilde{B}_y$ are again projective and unitary, one can find another unitary matrix $V_B$ that while leaving $\tilde{B}_0$ unchanged, acts on $\tilde{B}_1$ to give
    \begin{equation}
        V_B\tilde{B}_1V_B^\dagger = T_3\otimes\I\,,
    \end{equation}
    where $T_3$ is defined in \eqref{eq:t3-def}. We can then define the unitary matrix in \eqref{eq:u_b} as $U_B = V_BU'_{B}$.
    
    Having characterised Bob's observables, let us now look at Alice's observables. Here we follow a procedure similar to \cite{sarkarCertificationIncompatibleMeasurements2022}, wherein we restrict Bob's Hilbert space to $\C^3$ on which $Z$ and $T_3$ act and use the fact that these observables 
    are \emph{genuinely incompatible}. Genuine incompatibility implies non-existence of any common invariant subspace and we show in \autoref{app:GI} that $Z$ and $T_3$ are in fact genuinely incompatible using Burnside's theorem on irreducible matrix algebras. It is a straightforward conclusion that $\tilde{C}_0^*=(\alpha Z + \beta T_3)^*$ and $\tilde{C}_1^*=(\gamma Z + \delta T_3)^*$ are also genuinely incompatible.
    
    Obtaining the maximally entangled state as the one that achieves the Tsirelson bound and the characterisation of Alice's observables is a simultaneous process. To see how, observe that while working in the restricted Hilbert space, the state $\ket\psi$, which by assumption violates the Bell inequality maximally, is reduced to a mixed state $\Tr_{B'}(\ket\psi)=\rho_{AB}\in\h_A\otimes\h_B$, where $\h_B= \C^3$. As the maximal violation of the Bell inequality implies that the state $\rho_{AB}$ satisfies the nullifying conditions,
    \begin{equation}
        L_i \rho_{AB} = 0,
    \end{equation}
    all the eigenvectors of the state $\ket{\phi_i}$ also satisfy them. Therefore, we can consider the eigenvectors individually and use the following two facts for a pure state written in Schmidt-decomposed form, $\ket\phi = \lambda_l \ket{e_l}\ket{f_l}$.

    \textbf{Fact}: We can always write $\ket\phi$ as
    \begin{equation}
        \ket{\phi} = (M_A\otimes P_B)\ket{\phi_+},
    \end{equation}
    where $\ket{\phi_+}$ is the maximally entangled state, $M_A$ is a unitary matrix such that $M_A\ket{e_l} = \ket{f_l}^*$ and $P_B = \sqrt{d}\sum_i\lambda_i\op{f_i}$.

    \textbf{Fact}: For any two matrices $R$ and $Q$ and the maximally entangled state $\ket{\phi_+}$, 
    \begin{equation}
        R\otimes Q\ment = RQ^T\otimes \I \ment.
    \end{equation}

    Using the two facts above and the genuine incompatibility of $\tilde{C}_0^*$ $\tilde{C}_1^*$, we can show that there exists a unitary $U_A$ such that under the action of this unitary Alice's observables are
    \begin{align}
        U_A A_0 U_A^\dagger = (\alpha Z + \beta T_3)^*\otimes\I_{A'BB'} ,\\
        U_A A_1 U_A^\dagger = (\gamma Z + \delta T_3)^*\otimes\I_{A'BB'} .
    \end{align}

We see that the unitaries act on the Hilbert space by splitting it in $\C^3\otimes\C^3\otimes\h_{A'B'}$. Therefore, it is not surprising to find that the action of these unitaries on the state $\ket\psi$ is
    \begin{equation}
        U_A\otimes U_B\ket\psi = \ment\otimes\ket{\psi'},
    \end{equation}
    where $\ket{\psi'}$ is some state acting on $\h_{A'}\otimes\h_{B'}$,
    which completes the proof.
\end{proof}

\section{Classical Value of the Bell Inequality}
\label{sec:classical}
The classical value of the Bell inequality can be found by characterising the vertices of the local polytope. All the vertices can be written in the following form,
\begin{multline}
\beta_C = 
    -\frac{2\sin(3\theta_\alpha)\cos({\theta_\alpha+s_\alpha\frac{2\pi}{3}})}{\sin{3(\theta_\alpha-\theta_\beta)}} \\+ \frac{2\sin({3\theta_\beta})\cos({\theta_\beta+s_\beta\frac{2\pi}{3}})}{\sin{3(\theta_\alpha-\theta_\beta)}} \\-
    \frac{2\sin(3\theta_\gamma)\cos({\theta_\gamma+s_\gamma\frac{2\pi}{3}})}{\sin{3(\theta_\gamma-\theta_\delta)}} \\+ \frac{2\sin({3\theta_\delta})\cos({\theta_\delta+s_\delta\frac{2\pi}{3}})}{\sin{3(\theta_\gamma-\theta_\delta)}},
\end{multline}
where $s_\alpha,s_\beta\dots$ take on values $\{0,1,-1\}$ for different vertices.
Now to constrain the above value, it would be valuable to collect all the constraints on the Bell parameters that enable self-testing of the Bell inequality. For some integer $n\in\mathbb{Z}$:
\begin{align}
    c_1 :&\quad \theta_\alpha-\theta_\beta =\theta_\gamma-\theta_\delta + n \pi\,, \\
    c_2 :&\quad \theta_\alpha-\theta_\beta\neq \frac{n\pi}{3}\,,\ \theta_\alpha,\theta_\beta,\cdots \neq  \frac{n\pi}{3}, \\
    c_3 :&\quad \abs{\alpha}^2+ \abs{\beta}^2+ \abs{\gamma}^2+ \abs{\delta}^2 =2 \,, \\
    c_4 :&\quad\frac{\alpha^2}{\alpha^*} + \frac{\beta^2}{\beta^*}  = \frac{\gamma^2}{\gamma^*} + \frac{\delta^2}{\delta^*} = 1. \\ 
    c_5 :&\quad \alpha\beta^*+\gamma\delta^*.
\end{align}
For simplicity, let us take $n=0$ in the first constraint. The constraint $c_1$ lets us write $\theta_\gamma = \theta_\alpha-\theta_\beta+\theta_\delta$. 
The constraint $c_2$ prevents the classical value from being infinite. 
From $c_4$ we can extract the expressions for the absolute values of the Bell parameters (see \eqref{eq:thetacondi1}).
Notice in $c_3$ that one of the possible solutions for the constraint is obtained by choosing $\theta_\alpha-\theta_\beta = \frac{\pi}{6}$. Using which $c_5$ reduces to 
\begin{multline}
    \sin(6 \beta)+\sin(6 \delta) = 0 \\\implies \theta_\delta= \theta_\beta \pm (2 n +1) \frac{\pi}{6}\,, \forall n \in \mathbb{Z}.
\end{multline}
The maximum classical value, under all of the above assumptions is a function only of $\theta_\beta$, and in the range $0<\theta_\beta<\pi/6$ is given by:
\begin{equation}
\label{eq:classical-value}
    \beta_C = \begin{cases}
        3 \cos (2 \beta )+\cos (4 \beta ), & 0<\theta_\beta\leq\frac{\pi}{12}, \\
        3 \sin \left(2 \beta +\frac{\pi }{6}\right)+\sin \left(4 \beta -\frac{\pi }{6}\right), & \frac{\pi}{12}\leq\theta_\beta<\frac{\pi}{6}.
    \end{cases}
\end{equation}
See \autoref{fig:classical-value} for the plot of above expressions. Note that for this plot, we only considered a partial solution to the constraint $c_3$. Only the interval $0\leq\theta_\beta\leq\pi/6$ is plotted in the figure, as the maximum classical value is periodic. The minimum in this interval coincides with the maximum classical value of the SATWAP inequality for $d=3$, which is $\beta_C=3.09808$ when $\theta_\beta={\pi}/{12}$. The SATWAP inequality can be obtained from our inequality by further assuming $\theta_\alpha+\theta_\beta=0$ \cite{sarkarSelftestingQuantumSystems2021}. It is easy to see that in the whole region the maximum classical value in \eqref{eq:classical-value} is less than the Tsirelson bound of 4. The endpoints of the interval are excluded as they violate the constraints provided by self-testing, namely $c_2$. 

\begin{figure}
    \centering
    \includegraphics[width=0.88\linewidth]{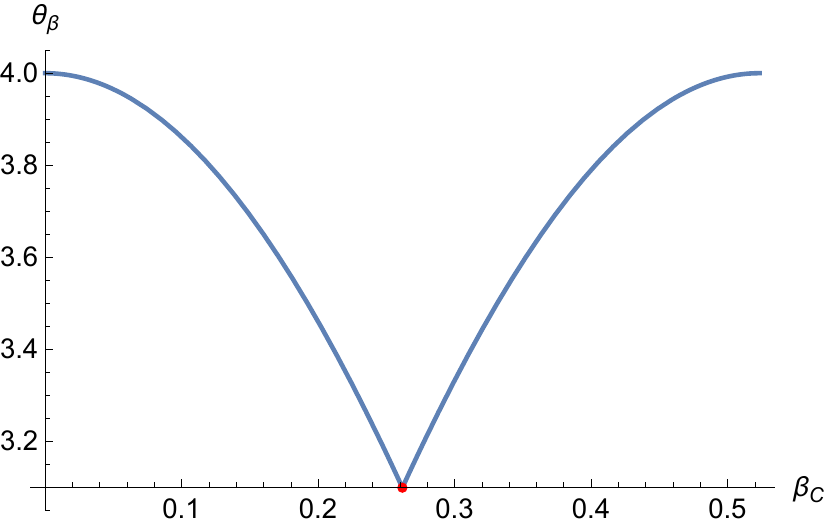}
    \caption{The plot depicts the maximum classical value under the constraints $\theta_\alpha = \theta_\beta +\pi/6$, $\theta_\gamma = \theta_\delta +\pi/6$ and $\theta_\delta = \theta_\beta +\pi/6$. The plot is for the interval $0\leq\theta_\beta\leq\pi/6$ as the maximum classical value is periodic in nature and the period is $\pi/6$. The red dot denotes the maximum classical value of the SATWAP inequality for $d=3$. }
    \label{fig:classical-value}
\end{figure}

\section{Conclusions}

In this article, we derived a self-testing statement for the Bell scenario for bipartite systems of local dimension $d=3$ considering the minimal number of measurement settings and outcomes per observer.
We started with the most general Bell operator inspired by the CGLMP Bell inequality. The imposition of  sum-of-squares decomposition combined with the projective and unitary three-outcome observables led to simple constraints on the parameters in the Bell inequality for the state to maximally violate the inequality when it is nullified by the elements of the sum-of-squares decomposition. From here, we determined the maximal quantum value achievable by our Bell operator to be 4. 

We proved analytically that the Bell inequality self-tests the maximally entangled state and a class of measurement observables that are constrained by a nominal number of constraints on the Bell parameters.
We also showed that the resulting class of self-testing Bell inequalities is a broader class of inequalities than previously known.  We also demonstrated that the SATWAP inequalities are a special case of our inequality. The fact that the SATWAP inequality for $d=3$ coincides with the point where the difference between the maximum classical value and the maximum quantum value is greatest, implies an optimality in the derived class of Bell inequalities.

\section{Acknowledgments}
This work was supported by the QuantERA II Programme (VERIqTAS project), which has received funding from the European Union’s Horizon 2020 research and innovation programme under Grant Agreement No. 101017733, and from the Polish National Science Centre (Grant No. 2021/03/Y/ST2/00175). We also acknowledge funding from the European Union’s Horizon Europe research and innovation programme under Grant Agreement No. 101080086 (NeQST). S.S. acknowledges support from the National Science Centre, Poland, under Grant Opus 25 (No. 2023/49/B/ST2/02468).

\bibliography{references}

\appendix
\section{Theorem 1: Sum-of-Squares and Tsirelson bound}
\label{app:theorem1}

The aim of \autoref{thm:theorem1} is to prove that there exists a sum-of-squares decomposition for the Bell operator,
\begin{multline}
    W=A_0\otimes \left(\alpha B_0+\beta B_1\right)+A_1\otimes \left(\gamma B_0+\delta B_1\right)\\
    +A_0^{\dagger}\otimes \left(\alpha B_0+\beta B_1\right)^{\dagger}+A_1^{\dagger}\otimes \left(\gamma B_0+\delta B_1\right)^{\dagger},
\end{multline}
and to find the conditions on the Bell parameters for when the SoS decomposition is valid.
Consider the sum-of-squares of $L_i$s defined in \eqref{eq:l1l2},
\begin{multline}
\label{eq:sos-app}
\frac{1}{2}\left(L_1^\dagger L_1 + L_1 L_1^\dagger+ L_2^\dagger L_2  + L_2 L_2^\dagger \right) \\= 2 \I - W + \frac{1}{2}\left(S_1 + S_2\right)\,,
\end{multline}
where
\begin{align}
    L_1 &= \I - A_0 \otimes (\alpha B_0 + \beta B_1)\,,\\
    L_2 &= \I - A_1 \otimes (\gamma B_0 + \delta B_1)\,,\\
    S_1 &= \I \otimes (\alpha B_0 + \beta B_1)^\dagger (\alpha B_0 + \beta B_1) \nonumber\\
    &+ \I \otimes (\alpha B_0 + \beta B_1) (\alpha B_0 + \beta B_1)^\dagger\,, \mbox{ and}\\
    S_2 &= \I \otimes (\gamma B_0 + \delta B_1)^\dagger  (\gamma B_0 + \delta B_1)\\
     &+ \I \otimes (\gamma B_0 + \delta B_1)  (\gamma B_0 + \delta B_1)^\dagger.
\end{align}
When we expand the products in $S_1$ and $S_2$, equation \eqref{eq:sos-app} becomes
\begin{multline}
\label{eq:sos2}
   \frac{1}{2}\left(L_1^\dagger L_1 + L_1 L_1^\dagger  + L_2^\dagger L_2 + L_2 L_2^\dagger \right)\\= (2 + \abs{\alpha}^2+ \abs{\beta}^2+ \abs{\gamma}^2+ \abs{\delta}^2) \I - W \\
   +\frac{\I}{2}\otimes (\alpha^*\beta + \gamma^*\delta)(B_0^\dagger B_1+B_1B_0^\dagger )\\+\frac{\I}{2}\otimes (\alpha\beta^* + \gamma\delta^*)(B_1^\dagger B_0+B_0 B_1^\dagger )\,.
\end{multline}
The cross terms in the above sum-of-squares vanish if we put 
\begin{equation}
\label{eq:ab-gd-sum-zero-app}
    (\alpha\beta^* + \gamma\delta^*) = (\alpha^*\beta + \gamma^*\delta)=0.
\end{equation}
Now, for the proof of the corollary, if we assume that a state $\ket\psi$ maximally violates the Bell inequality, then the following relation holds, whence $L_i$s are called the nullifiers of the state.
\begin{equation}
\label{eq:nullifiers-app}
    L_i\ket\psi = L_i^\dagger\ket\psi = 0 \qfor  i=1,2.
\end{equation}
The expectation value of the sum-of-squares,
\begin{multline}
\label{eq:ev-sos}
    \ev{\sum_{i=1}^2 (L_i^\dagger L_i+L_i L_i^\dagger)}{\psi}  \\= T_Q \ip{\psi} -\ev{W}{\psi} = 0.
\end{multline}
immediately gives us the Tsirelson bound or the quantum bound for this inequality as $T_Q = (2 + \abs{\alpha}^2+ \abs{\beta}^2+ \abs{\gamma}^2+ \abs{\delta}^2)$. 

For the second part of the corollary, let us examine \eqref{eq:nullifiers-app} a little more closely. $L_1\ket\psi = 0$ can be rewritten as
\begin{align}
    \label{eq:nullA0}
    A_0 \otimes (\alpha B_0 + \beta B_1) \ket\psi &= \ket\psi,\\
    \label{eq:b3isI}
    \I \otimes (\alpha B_0 + \beta B_1)^3 \ket\psi &= \ket\psi, \\
    \I \otimes (\alpha B_0 + \beta B_1) \ket\psi &= A_0^\dagger \otimes \I \ket\psi,
\end{align}
where the second equation is obtained from the first by multiplying by $L_1$ two times and the third equation is obtained by multiplying the first  by $A_0^\dagger$.
In addition, multiplying the third by $(\alpha B_0 + \beta B_1)^\dagger$ and using $L_1^\dagger\ket\psi=\ket\psi$.
\begin{multline}
    \I \otimes (\alpha B_0 + \beta B_1)^\dagger(\alpha B_0 + \beta B_1)\ket\psi 
    =\ket\psi.
\end{multline}

 Considering the action of these operators only on the system $B$, we have the following relations.
 \begin{subequations}
\label{eq:B-conditions}
 \begin{align}
 \label{eq:AB-B3=i}
 (\alpha B_0 +\beta B_1)^3 &= \I\,,\\
\label{eq:AB-B2=Bct}
 (\alpha B_0 +\beta B_1)^2 &= (\alpha B_0 +\beta B_1)^\dagger,\\
 \label{eq:AB-Bunitary}
 (\alpha B_0 +\beta B_1)^\dagger(\alpha B_0 +\beta B_1) &= \I\,.
 \end{align}
 \end{subequations}
 
Note here that similar relations will be obtained from the second nullifier $ L_2 = \I - A_1 \otimes (\gamma B_0 + \delta B_1)$, in terms of $\gamma$ and $\delta$.
\begin{subequations}
\label{eq:B-conditions-2}
\begin{align}
 (\gamma B_0 +\delta B_1)^3 &= \I\,,\\
 (\gamma B_0 +\delta B_1)^2 &= (\gamma B_0 +\delta B_1)^\dagger,\\
 \label{eq:GD-Bunitary}
 (\gamma B_0 +\delta B_1)^\dagger(\gamma B_0 +\delta B_1) &= \I\,.
\end{align}
\end{subequations}
These are the operators $C_0$ and $C_1$ in \eqref{eq:C-operators}, defined respectively as $\alpha B_0 +\beta B_1$ and $\gamma B_0 +\delta B_1$.
Now we are in a position to say something about the quantum bound. Expanding equations \eqref{eq:AB-Bunitary} and \eqref{eq:GD-Bunitary}:
\begin{align}
\label{eq:ab}
\left(\abs{\alpha}^2+\abs{\beta}^2-1\right)\I+\alpha\beta^*B_0B_1^{\dagger}+\alpha^*\beta B_1B_0^{\dagger}&=0, \\
\label{eq:gd}
\left(\abs{\gamma}^2+\abs{\delta}^2-1\right)\I+\gamma\delta^*B_0B_1^{\dagger}+\gamma^*\delta B_1B_0^{\dagger}&=0.
\end{align}
Adding the above two equations and using \eqref{eq:ab-gd-sum-zero-app} gives us the needed expression:
\begin{align}
    \left(\abs{\alpha}^2+\abs{\beta}^2 + \abs{\gamma}^2+\abs{\delta}^2-2\right) \I = 0, \\
    \implies \abs{\alpha}^2+ \abs{\beta}^2+ \abs{\gamma}^2+ \abs{\delta}^2 = 2.
\end{align}
Substituting the value of this expression in $T_Q$ from \eqref{eq:ev-sos} makes the quantum bound of the Bell operator $T_Q=4$, which completes the proof.

\section{Theorem 2: Self-testing statement}
\label{app:theorem2}
Now to characterise $B_0$ and $B_1$ as valid $3$-outcome measurements that satisfy the relations derived in the previous section, we consider two cases. First is the case where $B_0$ and $B_1$ act on three dimensional Hilbert space $\C_3$, and second, where the observables are acting on some arbitrary dimensional Hilbert space $\h_B$. Here, the utility of the restriction in the first case is only to provide us with a relation among the Bell parameters that is still valid in the second case.

\subsection{Three dimensional Hilbert space}
To start with the former, we use the fact that in the case of finite-dimensional Hilbert spaces, when $B_0$ and $B_1$ satisfy conditions \eqref{eq:operator-conditions}, without loss of generality, we can assume one of the operators, say $B_0=Z$, and then $B_1$ can be written in the general form
\begin{multline}
\label{eq:B1}
    B_1=aZ+bZ^2+X(c\I+dZ+eZ^2)\\+X^2(f\I+gZ+hZ^2),
\end{multline}
in terms of the elements of the three dimensional Heisenberg-Weyl basis. The matrices $Z$ and $X$ are also called the clock and shift matrices and are defined as $Z=\sum_{i=0}^2\omega^i\ket{i}\bra{i}$ and $X=\sum_{i=0}^2\ket{i}\bra{i+1}$. 
In recent works, \eg \cite{custom-bell-inequalities, Kaniewski2019maximalnonlocality, salavrakosBellInequalitiesTailored2017}, it is noted that the analysis of such a Bell scenario is simpler when the observables are defined in the Heisenberg-Weyl basis.
When we impose on the operators $B_0$ and $B_1$, the need to satisfy the latter two out of the three relations in \eqref{eq:B-conditions}, \ie
\begin{multline}\label{eq:condi1}
\left(|\alpha|^2+|\beta|^2-1\right)\I+\alpha\beta^*B_0B_1^{\dagger}+\alpha^*\beta B_1B_0^{\dagger}=0,
\end{multline}
and
\begin{multline}\label{eq:condi2}
(\alpha^2-\alpha^*)B_0^{\dagger}+(\beta^2-\beta^*)B_1^{\dagger}+\alpha\beta\left(B_0B_1+B_1B_0\right)=0.
\end{multline}
By substituting $B_0$ and $B_1$, we can compare the coefficients of the elements of the Heisenberg-Weyl basis to obtain the following relations between the coefficients in \eqref{eq:B1}.
\begin{subequations}
\label{eq:B1-parameters}
\begin{align}
a&=\frac{\abs{\alpha}^2+\abs{\beta}^2-1}{3 \alpha^*\beta}+\frac{\alpha^2-\alpha^*}{3\alpha\beta},\quad b=0,\\
c^*&=\frac{\alpha\beta}{\beta^2-\beta^*}\omega^2 h=-\frac{\alpha^*\beta}{\alpha\beta^*}\omega^2 h,\\
g^*&=\frac{\alpha\beta}{\beta^2-\beta^*} d=-\frac{\alpha^*\beta}{\alpha\beta^*} d,\\
e^*&=\frac{\alpha\beta}{\beta^2-\beta^*}\omega f=-\frac{\alpha^*\beta}{\alpha\beta^*}\omega  f. 
\end{align}
\end{subequations}

On the other hand, by substituting $B_0$ and $B_1$ in equations \eqref{eq:B-conditions-2}, we get the above relations in terms of parameters $\gamma$ and $\delta$. The following set of equations summarises the relations between the four Bell parameters from the above analysis.

\begin{equation}
\label{eq:bell-params-app}
\frac{\alpha\beta}{\beta^2-\beta^*}=-\frac{\alpha^*\beta}{\alpha\beta^*}=-\frac{\gamma^*\delta}{\gamma\delta^*}=\frac{\gamma\delta}{\delta^2-\delta^*}.
\end{equation}

Rearranging the terms in the first and last equality in \eqref{eq:bell-params-app}, we get
\begin{align}
\label{eq:alpha-square-is-1-app}
    \frac{\alpha^2}{\alpha^*} + \frac{\beta^2}{\beta^*} = 1 = \frac{\gamma^2}{\gamma^*} + \frac{\delta^2}{\delta^*} .
\end{align}
We can obtain additional expressions equivalent to \eqref{eq:bell-params-app} that are useful in later sections:
\begin{equation}
\label{eq:bell-params-app-2}
     \frac{\alpha^2-\alpha^*}{\alpha\beta} = \frac{\alpha\beta}{\beta^2-\beta^*} =\frac{\gamma^2-\gamma^*}{\gamma\delta}=\frac{\gamma\delta}{\delta^2-\delta^*}.
\end{equation}

The parameters in the Bell operator are complex and therefore can be written in the polar form: $\alpha=|\alpha|e^{\iota\theta_\alpha}$, 
$\beta=|\beta|e^{\iota\theta_\beta}$,
$\gamma=|\gamma|e^{\iota\theta_\gamma}$ and $\delta=|\delta|e^{\iota\theta_\delta}$. 
We can use the polar forms of the parameters in \eqref{eq:alpha-square-is-1-app} and individually compare the real part and the imaginary part with $1$ and $0$ respectively.
\begin{align}
    \abs{\alpha}e^{\iota 3 \theta_\alpha} + \abs{\beta}e^{\iota 3 \theta_\beta} &= 1,\\
    \qq*{Re:} \abs{\alpha}\cos{3\theta_\alpha} + \abs{\beta}\cos{3 \theta_\beta} &= 1,\\
    \qq*{Im:}\abs{\alpha}\sin{3\theta_\alpha} + \abs{\beta}\sin{3 \theta_\beta} &= 0.
\end{align}
A similar pair of relations for $\gamma$ and $\delta$ can be obtained as well, and is used to solve for the values of $\abs{\alpha}$, $\abs{\beta}$, $\abs{\gamma}$ and $\abs{\delta}$.
\begin{align}\label{eq:thetacondi1}
\abs{\alpha}=-\frac{\sin{3\theta_\beta}}{\sin{3(\theta_\alpha-\theta_\beta)}},& \quad \abs{\beta}=\frac{\sin{3\theta_\alpha}}{\sin{3(\theta_\alpha-\theta_\beta)}},\nonumber\\
\abs{\gamma}=-\frac{\sin{3\theta_\delta}}{\sin{3(\theta_\gamma-\theta_\delta)}},& \quad \abs{\delta}=\frac{\sin{3\theta_\gamma}}{\sin{3(\theta_\gamma-\theta_\delta)}}.
\end{align}
The absolute values do not go to infinity and are also not equal to zero, otherwise the Bell inequality would be trivial. Therefore, we have to ensure that
\begin{equation}
    \theta_{\alpha}-\theta_\beta=\theta_\gamma-\theta_\delta \neq n\frac{\pi}{3}\quad\quad \forall n \in \mathbb{Z},
\end{equation}
and for all integers $n_i$
\begin{align}
    \theta_i \neq \frac{n_i\pi}{3}\qcomma i\in\{\alpha,\beta,\gamma,\delta\}.
\end{align}
By substituting the polar form in \eqref{eq:bell-params-app}  
we get
\begin{eqnarray}\label{eq:thetacondi2}
\theta_{\alpha}-\theta_\beta=\theta_\gamma-\theta_\delta+n\pi\quad\quad \forall n\in integers
\end{eqnarray}
A relation that will be useful later on can be derived from \eqref{eq:alpha-square-is-1-app}, by multiplying with its complex conjugate,
\begin{align}
    \label{eq:alpha-sq-beta-sq-1}
    1-\abs{\alpha}^2 -\abs{\beta}^2 &=\frac{(\alpha{\beta^*})^2}{\alpha^*\beta} +\frac{({\alpha^*}{\beta})^2}{\alpha\beta^*}.
\end{align}

\subsection{Tracelessness of Bob's observables}
In the last section the Bell operator were taken by design to be traceless. In the case of a Hilbert space with some general dimension $d$, if the observables are shown to be traceless with eigenvalues $\{1,\omega,\omega^2\}$, then the $d$-dimensional Hilbert space splits into $\mathcal{H}_1\otimes\mathcal{H}_2$, such that $\mathcal{H}_1$ has dimension $3$. 

For ease of notation, we use the substitution
\begin{equation}
    \kappa=\frac{\gamma\delta^*}{\gamma^*\delta}=\frac{\alpha\beta^*}{\alpha^*\beta} .
\end{equation}
Using the polar form tells us that the modulus of $\kappa$ is 1,
$$\kappa=e^{2\iota(\theta_\alpha-\theta_\beta)}=e^{2\iota(\theta_\gamma-\theta_\delta)}.$$

First, taking the trace of \eqref{eq:condi2}, then, considering also the trace of the equation obtained by multiplying \eqref{eq:condi1} by $B_0^{\dagger}$ and $B_1^{\dagger}$, we have the following set of three equations.
\begin{subequations}
\begin{multline}
\frac{(\alpha^2-\alpha^*)}{2\alpha\beta}\Tr(B_0^{\dagger})+\frac{(\beta^2-\beta^*)}{2\alpha\beta}\Tr(B_1^{\dagger})\\=-\Tr(B_0B_1),    
\end{multline}
\begin{multline}
\frac{(\abs{\alpha}^2+\abs{\beta}^2-1)}{\alpha^*\beta}\Tr(B_0^{\dagger})+{\kappa}\Tr(B_1^{\dagger})\\=- \Tr(B_1B_0),
\end{multline}
\begin{multline}
    \frac{1}{\kappa} \Tr(B_0^{\dagger}) + \frac{(\abs{\alpha}^2+\abs{\beta}^2-1)}{\alpha\beta^*}\Tr(B_1^{\dagger})\\=-\Tr(B_0B_1).
\end{multline}
\end{subequations}

We can eliminate the term $\Tr(B_0B_1)$ and simplify the expressions using \eqref{eq:bell-params-app}, \eqref{eq:bell-params-app-2} and \eqref{eq:alpha-sq-beta-sq-1}. 
\begin{align}
    \left(\frac{1}{2\kappa} + \kappa^2  \right)\Tr(B_0^{\dagger}) &=  \frac{3}{2}\kappa \Tr(B_1^{\dagger}), \\
    \frac{3}{2\kappa}  \Tr(B_0^{\dagger}) &= \left(\frac{\kappa}{2}  + \frac{1}{\kappa^2} \right)\Tr(B_1^{\dagger}) .
\end{align}

which result in,
\begin{multline}
    \frac{\left(\kappa^3 -1\right)^2}{\kappa^2} \Tr(B_1^\dagger) = 0, \\\qand* \frac{\left(\kappa^3 -1\right)^2}{\kappa^4} \Tr(B_0^\dagger) = 0.
\end{multline}
Using the definition of $\kappa$, we find that the observables are traceless when:
\begin{align}
    \exp{\iota 6(\theta_\alpha -\theta_\beta)} \neq 1 &\implies (\theta_\alpha -\theta_\beta) \neq \frac{ n \pi}{3} \forall n \in \mathbb{Z}, \\
    \exp{\iota 6(\theta_\gamma -\theta_\delta)} \neq 1 &\implies (\theta_\gamma -\theta_\delta) \neq \frac{ n \pi}{3} \forall n \in \mathbb{Z}.
\end{align}
Thus, to restrict the dimension of the observables to $3\times n$ where $n$ is any positive integer, the two conditions above have to be satisfied.

\subsection{Characterising Bob's observables}
Now that the dimensions of the Hilbert space are restricted to $3\times n$, using the unitary freedom, we can in general choose a unitary ${U'}_B$ so that $\tilde{B}_0={U'}_B B_0 {U'}_B^\dagger=Z\otimes\I$.
Then we can write $\tilde{B}_1={U'}_BB_1{U'}_B^\dagger=\sum_{i,j=0}^{2}\ket{i}\bra{j}\otimes F_{ij}$, where our task now is to determine the matrices $F_{ij}$. Substituting $\ti{B}_0$ and $\ti{B}_1$ in \eqref{eq:condi2},
\begin{multline}\label{eq:B1condi1}
(\alpha^2-\alpha^*)\sum_{i=0}^2\omega^{-i}\op{i}\otimes\I+(\beta^2-\beta^*)\sum_{i,j=0}^{2}\op{i}{j}\otimes F_{ji}^{\dagger}\\+\alpha\beta\sum_{i,j=0}^{2}(\omega^i+\omega^j)\op{i}{j}\otimes F_{ij}=0.
\end{multline}
Consider first the diagonal elements of the above matrix equation. To solve for $F_{ii}$, we set $j=i$ and disregard the 3-dimensional subsystem and divide by $\alpha\beta$. The diagonal elements and their conjugate transpose are as follows.
\begin{align}
\label{eq:B1condi2}
\frac{(\alpha^2-\alpha^*)}{\alpha\beta}\omega^{-i}\I+\frac{(\beta^2-\beta^*)}{\alpha\beta}F_{ii}^{\dagger}+2\omega^i F_{ii}=0,\\
\frac{(\alpha^2-\alpha^*)^*}{\alpha^*\beta^*}\omega^{i}\I+2\omega^{-i} F_{ii}^\dagger +\frac{(\beta^2-\beta^*)^*}{\alpha^*\beta^*}F_{ii}=0.
\end{align}
Using \eqref{eq:bell-params-app} and \eqref{eq:bell-params-app-2} we eliminate $F_{ii}^\dagger$ to get,
 \begin{eqnarray}\label{eq:Fii1}
F_{ii}=\left(\kappa^2 +\frac{2}{\kappa}\right)\frac{\omega^i}{3}\I.
\end{eqnarray}
Looking at the off-diagonal case, where $i\neq j$, from \eqref{eq:B1condi1}
\begin{align}
-(\beta^2-\beta^*)F_{ji}^{\dagger}&=\alpha\beta(\omega^i+\omega^j)F_{ij}=-\alpha\beta\omega^{-(i+j)}F_{ij} \nonumber\\
\label{eq:Fijcondi1}
\implies&  -\kappa F_{ji}^{\dagger}=\omega^{-(i+j)}F_{ij},
\end{align}
where we use that fact that for the cube root of unity $\omega = \exp(\iota 2 \pi/3)$, 
\begin{equation}
    \omega^i+\omega^j = \omega^{-(i+j)},\ \mbox{for } i,j=0,1,2.
\end{equation}

Since $\tilde{B}_0$ and $\tilde{B}_1$ are unitary observables, they satisfy $\tilde{B}_1^{\dagger}\tilde{B}_1=\I$ and $\tilde{B}_1^2=\tilde{B}_1^{\dagger}$ which sets the following constraints on the matrices $F_{ij}$
\begin{align}\label{eq:Unicondi1}
\sum_{l=0}^2F_{il}F_{jl}^{\dagger}&=\I\delta_{ij}, \\
\label{eq:Unicondi2}
\qand*\sum_{l=0}^2F_{il}F_{lj}&=F_{ji}^{\dagger}.
\end{align}
Setting $i,j=0$ in \eqref{eq:Unicondi1} and \eqref{eq:Unicondi2} and rearranging,
\begin{align}
F_{01}F_{01}^{\dagger}+F_{02}F_{02}^{\dagger}&=\I -F_{00}F_{00}^{\dagger}, \\
F_{01}F_{10}+F_{02}F_{20}&=F_{00}^{\dagger}-F_{00}F_{00}.
\end{align}
Substituting the expressions for $F_{00}$ and $F_{00}^\dagger$ and using \eqref{eq:Fijcondi1}, we can eliminate $F_{02}F_{02}^\dagger$ and then $F_{01}F_{01}^\dagger$, 

\begin{multline}
   F_{01}F_{01}^\dagger = F_{02}F_{02}^\dagger = \frac{1}{9} \left(2- \left(\kappa^3+\frac{1}{\kappa^3} \right)\right) \I \\= \frac{2}{9}\left( 1-\cos{6(\theta_\alpha-\theta_\beta)}\right)\I.
\end{multline}

Similar relations for the matrices $F_{ij}$ with different values of $i,j$ can be worked out but it is not strictly necessary. We can once again choose an appropriate unitary matrix $V_B = U_0\oplus U_1\oplus U_2$, such that $U_i$s are unitary and it leaves $\tilde{B}_0$ invariant, \ie $V_B\tilde{B}_0V_B^\dagger=\tilde{B}_0$. If we take 
\begin{eqnarray}
U_0=c_0\I_{B'},\quad U_1= c_1F_{01},\quad U_2=c_2F_{02},
\end{eqnarray}
where, $c_i$s are some complex numbers, then by unitarity and because $V_B$ is a block diagonal matrix, we get the following constraints 
\begin{align}
    U_0U_0^\dagger &= \abs{c_0}^2\I = \I, \\
    U_1U_1^\dagger &= \abs{c_1}^2F_{01}F_{01}^\dagger = \I, \qand\\
    U_1U_1^\dagger &= \abs{c_2}^2F_{02}F_{02}^\dagger = \I.
\end{align}
which leads to the absolute values,
\begin{equation}
\label{eq:abs-values-cis}
    \abs{c_0}^2 = 1\qcomma\abs{c_1}^2 = \abs{c_2}^2  = \frac{9}{ \left(2- \left(\kappa^3+\frac{1}{\kappa^3} \right)\right)}.
\end{equation}

Under this unitary action we define $\ti{B}_1\gets V_B\ti{B}_1V_B^\dagger$, so that the constituent block matrices are $\ti{F}_{ij} = U_iF_{ij}U_j^\dagger$.
It follows immediately that $\tilde{F}_{ii}=F_{ii}$. 
Thus, for the diagonal elements, $\tilde{F}_{ii} = F_{ii}$. The off-diagonal elements are determined below using \eqref{eq:Fijcondi1}, \eqref{eq:Unicondi1} and \eqref{eq:Unicondi2}.
\begin{subequations}
\begin{align}
    \tilde{F}_{01} =& U_0F_{01}U_1^\dagger  = c_0c_1^*F_{01}F_{01}^\dagger = \frac{c_0c_1^*}{\abs{c_1}^2}\I,\\
    \tilde{F}_{10} =& U_1F_{10}U_0^\dagger  = -c_0^*c_1F_{01}F_{10} = -\kappa\omega\frac{c_0^*c_1}{\abs{c_1}^2}\I,\\
    \tilde{F}_{02} =& U_0F_{02}U_2^\dagger  = c_0c_2^*F_{02}F_{02}^\dagger = \frac{c_0c_2^*}{\abs{c_2}^2}\I,\\
    \tilde{F}_{20} =& U_2F_{20}U_0^\dagger  = -c_2F_{02}F_{20} = -\kappa\omega^2\frac{c_0^*c_2}{\abs{c_2}^2}\I.
\end{align}
\end{subequations}
To obtain $\tilde{F}_{12}$ and $\tilde{F}_{21}$, respectively put $i=1, j=0$ and $i=2, j=0$ in \eqref{eq:Unicondi1}, then multiply the first from left by $F_{01}$ and the latter by $F_{02}$. Solving for the requisite terms in both gives us
\begin{subequations}
\begin{multline}
    \tilde{F}_{12} = U_1F_{12}U_2^\dagger = c_1c_2^*F_{01}F_{12}F_{02}^\dagger \\= \frac{\omega}{3}\frac{c_1c_2^*}{\abs{c_2}^2}\left( \kappa^2-\frac{1}{\kappa}\right)\I,
\end{multline}
\begin{multline}
    \tilde{F}_{21} = U_2F_{21}U_1^\dagger = c_1^*c_2F_{02}F_{21}F_{01}^\dagger \\= \frac{\omega^2}{3}\frac{c_1^*c_2}{\abs{c_1}^2}\left( \kappa^2-\frac{1}{\kappa}\right)\I.
\end{multline}
\end{subequations}
As the matrices $\ti{F}_{ij}$ are all proportional to identity, we find that $\ti B_1=T_3\otimes\I_{B'}$, where
\begin{multline}
T_3 = \\ \begin{pmatrix}
    \frac{1}{3}\left( \kappa^2 +\frac{2}{\kappa}\right) & \frac{c_0}{{c_1}} & \frac{c_0}{{c_2}}\\
    -\kappa\omega\frac{c_0^*}{{c_1}^*} & \frac{\omega}{3}\left( \kappa^2 +\frac{2}{\kappa}\right) & \frac{\omega c_1}{3{c_2}}\left( \kappa^2-\frac{1}{\kappa}\right) \\
    -\kappa\omega^2\frac{c_0^*}{{c_2}^*} & \frac{\omega^2c_2}{3{c_1}}\left( \kappa^2-\frac{1}{\kappa}\right) & \frac{\omega^2}{3}\left( \kappa^2 +\frac{2}{\kappa}\right)
\end{pmatrix}.
\end{multline}

Observe that the denominator of the absolute values of $c_1$ and $c_2$ in \eqref{eq:abs-values-cis} can be expressed as a square, 
\begin{multline}
     \left(2- \left(\kappa^3+\frac{1}{\kappa^3} \right)\right) = - \left(\kappa^{3/2}-\frac{1}{\kappa^{3/2}}\right)^2 \\ = -(2 \iota \sin{3\theta_\kappa})^2,
\end{multline}
and therefore we can assume, without loss of generality,
\begin{equation}
c_1 = \frac{3 x}{\left( \kappa^{3/2} - \frac{1}{\kappa^{3/2}}\right)}\qand c_2 = \frac{3 x \omega^n}{\left( \kappa^{3/2} - \frac{1}{\kappa^{3/2}}\right)},
\end{equation}
where $x\in\mathbb{C}$ such that $\abs{x}^2=1$ and $n\in\mathbb{Z}$. Furthermore, if we take $x= \kappa^{-1/2}$, then $c_1 = 3 (\kappa^2 -\kappa^{-1})^{-1}$ and similarly for $c_2$.
The matrix $T_3$ becomes,
\begin{multline}
\label{eq:t3-def-app}
    T_3 = \frac{1}{3}\left[\left( \kappa^2+\frac{2}{\kappa} \right) Z + \left( \kappa^2-\frac{1}{\kappa} \right) ZX \right.\\
    \left.+ \omega\left( \kappa^2-\frac{1}{\kappa} \right) ZX^2\right].
\end{multline} 

Therefore, the observables acting on Bob's Hilbert space are $\ti B_0=Z_B\otimes\I_{B'}$ and $\ti B_1=T_B\otimes\I_{B'}$ under the unitary $U_B = V_B{U'}_B$. Now, because of the tensor product structure, the two observables  are not \emph{genuinely incompatible}, meaning that they share some common invariant subspace. On the other hand, if we restrict ourselves to the three dimensional subspace of $Z$ and $T_3$, we can employ Burnside's theorem for matrix algebras and show that they generate the same irreducible matrix algebra $\mathcal{A}$  generated by the matrices $Z$ and $X$, \ie $\mathcal{A}(\ti B_0, \ti B_1) = \mathcal{A}(Z, X)$.
This proof can be seen in \autoref{app:GI}.

Furthermore, one could also prove that the observables $\tilde{C}_0 = \alpha \ti B_0 + \beta \ti B_1$ and $\tilde{C}_1 = \gamma \ti B_0 + \delta \ti B_1$ are genuinely incompatible by using the relation $\alpha\beta^* +\gamma\delta^*=0$. With this information in hand, we can move on to Alice's observables.

\subsection{Alice's observables and the State}
In this section, we will show that there exists a unitary $U_A$ that plays the same role as $U_B$ of splitting the local Hilbert space and in the process we also show that the state in question that violates the Bell inequality maximally is indeed the maximally entangled state.

Let's start from the nullifiers, $L_1$ and $L_2$ from the SoS decomposition in \eqref{eq:nullifiers} with $\ti B_0 = Z\otimes\I$ and $\ti B_1 = T\otimes\I$,
\begin{align}
    A_0 \otimes (\alpha Z + \beta T)\otimes\I \ket\psi &= \ket\psi,\\
    A_1 \otimes (\gamma Z + \delta T)\otimes\I \ket\psi &= \ket\psi.
\end{align}
If we restrict the Bob's Hilbert space to the 3 dimensional subspace on which $Z$ and $T$ act, then the relations above are  satisfied by the state $\rho_{AB} = \Tr_{B'}(\op{\psi})$,
\begin{align}
    A_0 \otimes (\alpha Z + \beta T) \rho_{AB} &= \rho_{AB},\\
    A_1 \otimes (\gamma Z + \delta T) \rho_{AB} &= \rho_{AB}.
\end{align}
We can then use the eigen-decomposition of $\rho_{AB}$ to see that for each of it's eigenvectors, $\ket\phi_i$, the above equations still hold, \ie
\begin{align}
    A_0 \otimes (\alpha Z + \beta T) \ket\phi_i &= \ket\phi_i,\\
    A_1 \otimes (\gamma Z + \delta T) \ket\phi_i &= \ket\phi_i.
\end{align}
Multiplying both sides by the hermitian conjugate of observable acting on the Bob's Hilbert space,
\begin{align}
\label{eq:a0}
A_0 \otimes \I \ket\phi_i &= \I \otimes (\alpha Z + \beta T)^\dagger\ket\phi_i, \\
\label{eq:a1}
A_1 \otimes \I \ket\phi_i &= \I \otimes (\gamma Z + \delta T)^\dagger\ket\phi_i.
\end{align}

Now we use the fact that for a pure state with the Schmidt decomposition $\ket\phi_i = \sum_l \lambda_l \ket{e_l}\ket{f_l}$, one can always write
\begin{equation}
    \ket\phi_i = (U_A\otimes P_B) \ment
\end{equation}
where $U_A$ is a unitary that satisfies $U_A\ket{e_i} = \ket{f_i}^*$ for all $\ket{e_i}$ and $P_B = \sqrt{d}\sum_i\lambda_i\op{f_i}$ is a positive matrix. Let us substitute this in \eqref{eq:a0} with $\ti C_0=(\alpha Z + \beta T)$ and $\ti C_1=(\gamma Z + \delta T)$,
\begin{align}
    A_0 U_A\otimes P_B \ment &=  U_A \otimes \tilde{C}_0^\dagger P_B \ment, \\
    A_1 U_A\otimes P_B \ment &=  U_A \otimes \tilde{C}_1^\dagger P_B \ment.
\end{align}
Now using the fact that for any two matrices $R$ and $Q$ acting on the maximally entangled state,  $$R\otimes Q\ment = RQ^T\otimes \I \ment,$$
so we can move all the observables acting on Bob's Hilbert space to Alice's space:
\begin{align}
\label{eq:b1conj}
    \tilde{A}_0 P_B\otimes \I \ment &=  P_B \tilde{C}_0^*\otimes \I  \ment, \\
    \tilde{A}_1 P_B\otimes \I \ment &=  P_B \tilde{C}_1^* \otimes \I \ment.
\end{align}
Here $\tilde{A}_x = U_A^\dagger A_x U_A$. We can now compare the observables acting on the local Hilbert space of Alice and write the following equality
\begin{align}
\label{eq:a0-and-a1-app}
    \tilde{A}_0 P_B =  P_B \tilde{C}_0^* \qand 
    \tilde{A}_1 P_B =  P_B \tilde{C}_1^* .
\end{align}
Multiplying by their hermitian conjugate from the left leaves us with 
\begin{align}
    P_B \tilde{A}_0^\dagger\tilde{A}_0 P_B =P_B^2=  \tilde{C}_0^T P_B^2 \tilde{C}_0^* ,\\
    P_B \tilde{A}_1^\dagger\tilde{A}_1 P_B =P_B^2=  \tilde{C}_1^T P_B^2 \tilde{C}_1^* .
\end{align}

Rearranging the above equations to
\begin{align}
    \tilde{C}_0^*P_B^2=   P_B^2 \tilde{C}_0^* \qand
   \tilde{C}_1^* P_B^2=   P_B^2 \tilde{C}_1^* .
\end{align}
gives us the commutation relations $\comm{\tilde{C}_0^*}{P_B^2}=0$ and $\comm{\tilde{C}_1^*}{P_B^2}=0$. Due to the fact that $P_B$ is a positive matrix, the commutation relations $\comm{\ti{C}_1^*}{P_B}=0$ and $\comm{\ti{C}_0^*}{P_B}=0$ also hold.
It was shown in \cite{sarkarCertificationIncompatibleMeasurements2022} that for such a positive matrix $P_B$ and a set of genuinely incompatible observables $\{\ti{C}_i\}$, if $\comm{\ti{C}_i}{P_B}=0$, then the matrix $P_B$ is of the form $\lambda \I$ for some constant $\lambda$.
Therefore, in our case, where we have already shown that $\ti{C}_0^*$ and $\ti{C}_1^*$ are genuinely incompatible, then $P_B$ is of proportional to identity and we can rewrite \eqref{eq:a0-and-a1-app},

\begin{align}
\tilde{A}_0 = \tilde{C}_0^* \qand \tilde{A}_1 = \tilde{C}_1^*.
\end{align}
Observe that the unitary $U_A$ like $U_B$ splits the Hilbert space on Alice's side as $\mathbb{C}^3 \otimes \mathcal{H}^A_d$ 
where  $\mathcal{H}^A_{d'}$ is $d'$-dimensional Hilbert space. The unitaries $U_A$ and $U_B$ combined have the effect of splitting the whole Hilbert space as $\mathbb{C}^A_3\otimes\mathbb{C}^B_3 \otimes \mathcal{H}^A_{d'} \otimes \mathcal{H}^B_{d'}$. In the process, we have also shown, that the maximally entangled state satisfies the relations, as under the unitaries we have $U_A\otimes U_B\ket\psi = \ment\otimes\ket{\psi'}$, so that
\begin{equation}
   U_{AB} L_iU_{AB}^\dagger\ment\otimes\ket{\psi'} = 0, \mbox{ for } i=1,2,
\end{equation}
where we have used $U_{AB} = U_A\otimes U_B$ and $\ket{\psi'}$ is some state in $\mathcal{H}^A_{d'} \otimes \mathcal{H}^B_{d'}$. 
Therefore, $\ment$ violates the Bell inequality maximally. 

\section{Genuine Incompatibility of two observables}
\label{app:GI}
Given two matrices, $A$ and $B$, they are said to be Genuinely incompatible, if and only if they do not share any common invariant subspace. As mentioned earlier, we can use the Burnside's theorem on matrix algebras to show that the matrix algebra generated by $A$ and $B$ is irreducible.  In this case we are concerned with the matrices $B_0 = Z$ and $B_1=T$ (defined in \eqref{eq:t3-def}), where $T$ is a linear combination of some of the elements of the set of Weyl-Heisenberg operators which is generated by the matrix algebra $\mathcal{A}(Z,X)$.

The algebra, $\mathcal{A}(Z,X)$, is an irreducible algebra \ie an algebra of all $3\times3$ complex matrices, containing $9$ linearly independent elements.
Therefore, one can show that the algebra $\mathcal{A}(B_0,B_1)$ is irreducible, if it is the same as $\mathcal{A}(Z,X)$. 
\begin{lemma}
    The matrix algebra $\mathcal{A}(B_0,B_1)$ is irreducible.
\end{lemma}
\begin{proof}
It is enough to show that one can recover the matrix $X$ from $\mathcal{A}(B_0,B_1)$.
Let $B_0 = Z$ and to simplify things let's write $B_1=aZ+bZX+\omega b ZX^2$. Now, $B_0$ and $B_1$ are in the algebra by definition. We also have $\I = B_0^3$ as an element. Then, by using elementary operations the following matrices are also in the algebra.
\begin{align}
    M_1 &=\frac{1}{b}Z^2( B_1-aB_0) =  X+ \omega X^2, \\
    M_2 &= B_0^2 = Z^2, \\
    M_3 &= \frac{1}{3}(\I + Z + Z^2) = \mbox{diag}(1,0,0),
    \end{align}
    \begin{align}
    M_4 &= \frac{1}{3}(\I+\frac{Z}{\omega} + \frac{Z^2}{\omega^2})=\mbox{diag}(0,1,0), \\
    M_5 &= \frac{1}{3}(\I+\frac{Z}{\omega^2} + \frac{Z^2}{\omega})=\mbox{diag}(0,0,1), \\
    M_6 & = M_4 M_1^2 M_5 +\frac{M_3M_1^2M_4}{\omega} +  \frac{M_5 M_1^2 M_3}{\omega^2}  = X \,.
 \end{align}
 Now that we have recovered $X$, together with $B_0=Z$ we can generate the complete algebra.
 \end{proof}

\end{document}